\documentclass[reqno,centertags, 12pt]{amsart}
\usepackage{amsmath,amsthm,amscd,amssymb}
\usepackage{esint}  
\usepackage{latexsym}
\usepackage{graphicx}
\usepackage{enumitem}
\usepackage[mathscr]{eucal}
\usepackage{float}

\newcommand{\bbN}{{\mathbb{N}}}

\newcommand{\bbR}{{\mathbb{R}}}

\newcommand{\bbZ}{{\mathbb{Z}}}

\newcommand{\lb}{\label}

\newcommand{\supp}{\text{\rm{supp}}}

\newcommand{\bi}{\bibitem}

\newcommand{\beq}{\begin{equation}}
\newcommand{\eeq}{\end{equation}}
\newcommand{\ba}{\begin{align}}
\newcommand{\ea}{\end{align}}

\newcounter{smalllist}

\newcommand{\comm}[1]{}

\allowdisplaybreaks
\numberwithin{equation}{section}

\newtheorem{theorem}{Theorem}[section]
\newtheorem{proposition}[theorem]{Proposition}
\newtheorem{lemma}[theorem]{Lemma}
\newtheorem{corollary}[theorem]{Corollary}
\theoremstyle{definition}
\newtheorem{example}[theorem]{Example}

\newtheorem*{remark}{Remark}
\newtheorem*{remarks}{Remarks}

\newcommand{\jap}[1]{\langle #1 \rangle}
\newcommand{\Norm}[1]{\lVert#1\rVert}

\begin{document}

\title[Three Coauthors]{A Tale of Three Coauthors}
\author[B.~Simon]{Barry Simon$^{1,2}$}

\thanks{$^1$ Departments of Mathematics and Physics, Mathematics 253-37, California Institute of Technology, Pasadena, CA 91125.
E-mail: bsimon@caltech.edu}

\thanks{$^2$ Research supported in part by Israeli BSF Grant No. 2020027.}

\

\date{\today}

\begin{abstract}  We tell the story of the discovery of an interesting bound on finite sums and its application to comparison of Ising models.
\end{abstract}

\maketitle

\section{Introduction} \lb{s1}

In 2022, we celebrated Elliott Lieb's $90^{th}$ birthday.  On Friday, Jan 14, 2022, I had a draft of a single authored paper intended for a Festschrift to be published for Lieb.  Six days later, that paper had three coauthors who I hadn't met before Jan 14, 2022 (indeed, even now, I've only met them on Zoom).  This paper will explain the interesting story, expose some underlying machinery and sketch the proof of a lovely inequality on certain finite sums.  It will include an improvement of $50$ year old bounds of Griffiths~\cite{GriffTrick} comparing transition temperatures on generalized Ising models for different spins.  Because its a fun story and involves a charming inequality, I've given several talks on the material including at the Seminar on Analysis, Differential Equations and Mathematical Physics~\cite{SimonTalk}, an online seminar sponsored by the Institute of Mathematics, Mechanics and Computer Sciences of Southern Federal University, Rostov-on-Don, Russia.  After the talk, I was invited to by one of the organizers, Alexey Karapetyants, to contribute an article telling this story for the special $50^{th}$ anniversary issue of the Journal of Mathematical Sciences.  I was originally reluctant but persuaded by the fact that I might convince young workers to react properly to rejections since a paper rejection is part of the story.

I am writing a book for Cambridge Press entitled \emph{Phase Transitions in the Theory of Lattice Gases}~\cite{PTLG}.  It is, in many ways, the successor to my 1993 book \emph{The Statistical Mechanics of Lattice Gases}, Vol. I~\cite{SMLG}, from Princeton University Press.  That earlier book was mainly framework and largely left out all the most fun and beautiful elements of the theory: Correlation Inequalities, Lee-Yang, Peierls' Argument, BKT transitions, Infrared Bounds and Random Clusters and Currents  which are the subjects of the new book.  But since I decided to use a different publisher, this is certainly NOT volume 2 of the earlier work.  The main mathematical focus of the story is an inequality which one can state and admire without knowing anything about Ising models so I will not bother to define what exactly what they are instead referring the reader in general to the two books or to the paper with the three coauthors~\cite{MSW}, giving more precise references at appropriate points.

For our discussion we will only need to consider generalized Ising models in finite volume, a subset, $\Lambda$, of the lattice $\bbZ^\nu$.  For each site $j \in\Lambda$, one has a real valued ``spin'', $\sigma_j$.  For $A\subset\Lambda$, one defines
\begin{equation}\label{1.1}
  \sigma^A = \prod_{j\in A}\sigma^j
\end{equation}
One object needed to define the Gibbs measure is a function, $J:2^\Lambda\to\bbR$ called a coupling which is called \emph{ferromagnetic} if $J(A)\ge 0$ for all $A$.  One then forms a Hamiltonian,
\begin{equation}\label{1.2}
  h=-\sum_{A\subset\Lambda} J(A)\sigma^A
\end{equation}
The name ferromagnetic comes from the fact that when $J(A)\ge 0$, states with more spins parallel have lower energies and so according to Gibbs rules higher weights.

The other object one needs describes the uncoupled state which gives the $\{\sigma_j\}_{j\in\Lambda}$ the distribution of independent, identically distributed random variables with some common distribution, $d\mu$, a probability measure on $\bbR$ called the \emph{apriori measure}, which we will always suppose to be even.  One then fixes $J(\cdot)$ but varies $\mu$ and forms Gibbs states, $\jap{\cdot}_{\mu,\Lambda}$, according to the standard prescription~\cite[Chap. III]{SMLG}.

There are special choices of apriori measure that particularly concern us beginning with the spin-$\tfrac{1}{2}$ measure:
\begin{equation}\label{1.3}
  d\tilde{\mu}_{S=\tfrac{1}{2}} = \frac{1}{2}(\delta_{1}+\delta_{-1})
\end{equation}
More generally, for any $T>0$, we consider
\begin{equation}\label{1.4}
  db_T = \frac{1}{2}(\delta_{T}+\delta_{-T})
\end{equation}
$b$ is for Bernoulli.  Finally, for $S=\tfrac{1}{2}, 1, \tfrac{3}{2}, 2,\dots$, we consider the measure $d\tilde{\mu}_S$ which takes $2S+1$ values equally spaced between $-1$ and $1$, each with weight $1/(2S+1)$.  This is a scaled version of what is called spin $S$ Ising (which has maximum spin value $2S$).

\emph{Conflict of Interest Statement}: This is to certify that the author has no conflicts of interest.

\emph{Data Availability Statement}: There is no data associated with this paper.

\section{Wells Ordering} \lb{s2}

As I began to write about correlation inequalities, I wondered about a natural question.  Say that an apriori measure, $\nu$, on $\bbR$ \emph{Ising dominates} another measure $\mu$ if and only if for all $J(A)\ge 0$ and all $B$, one has that

\begin{equation}\label{2.1}
  \jap{\sigma^B}_{\mu,\Lambda} \le \jap{\sigma^B}_{\nu,\Lambda}
\end{equation}

In particular, for general $\mu$ compact support, does one have that $\mu$ Ising dominates $b_{T_-}$ and is Ising dominated by $b_{T_+}$ for suitable $0< T_- < T_+ <\infty$.  That would imply phase transitions occur for one apriori measure if and only if they do for all and inequalities on transition temperatures. To be explicit, if $\mu$ Ising dominates $b_{T_-(\mu)}$, and if $T_c(\mu)$ is a transition temperature for some fixed ferromagnetic pair interaction, one easily sees that

\begin{equation}\label{2.2}
  T_c(\mu) \ge T_-(\mu)^2 T_c(\text{\textit{classical Ising}})
\end{equation}

For most even minor aspects of the subject of correlation inequalities, there are several papers, sometimes as many as a dozen.  So I was surprised that I was unable to find a single published paper on the subject of what I just called Ising domination! Of course, it was unclear how to search for the subject in Google.  Eventually, I did find one paper of van Beijeren and Sylvester~\cite{vBS} that is unsatisfactory in that, in their theory, the analog of what I call $T_-$ is $0$ if $0\in\supp(\mu)$.  And I did find an appendix of a paper on another subject but that gets ahead of my story.

One of the pleasant things about writing a book on a subject that I once knew more about is that I get to rediscover things I've forgotten.  With the question of Ising domination in the back of my mind, I found an interesting footnote in a 1980 paper of Aizenman and er, B. Simon~\cite{AiSiRotor} entitled \emph{A comparison of plane rotor and Ising models}.  The footnote said

\medskip

\noindent\emph{then by results of Wells (D. Wells, \emph{Some moment inequalities for general spin Ising ferromagnets}, Indiana Univ. preprint)}
\begin{equation}\label{2.3}
   \jap{s_js_k}_{\beta,1} \le 2\jap{\sigma_j^{(1)}\sigma_k^{(1)}}_{\beta,2}
\end{equation}

\medskip

Here the left hand side is an Ising expectation and the right with the apriori measure of the $2D$ rotor with only couplings of the $1$ components.  So this was part of what seems to be an Ising domination result (the $2$ indicates the Ising measure should really be $b_{1/\sqrt{2}}$).

So I set about finding this preprint.  Google didn't help directly but did point me to a 1984 paper of Chuck Newman~\cite{NewWells} that mentioned Wells' Indiana University PhD. thesis.  So I wrote to Michael asking if he knew anything about our footnote and cced Chuck (who had been a grad student with me at Princeton) because I conjectured Wells had been his student.  Chuck replied and said he remembered that Wells had been Slim Sherman's student.  Sherman, the S of GKS and GHS was delightful character, long dead.  I then wrote to Kevin Pilgrim, the chair at Indiana, who located a copy of Wells thesis~\cite{WellsTh} for me on Proquest.  But he had no luck on the preprint nor on locating Wells through Indiana University alumni records (to get way ahead of the story, when I eventually reached Wells, he was surprised by this remarking ``gee, they don't have trouble finding me to send fund raising letters'')!  While the thesis did not have anything directly about the above inequality, it did have a general framework on what I called the Ising domination problem, lovely material that should have been published.

Given the beauty and relevance of the work in Wells thesis, I decided to include it in my book and to further develop it.   Wells exploited tools that Ginibre had introduced~\cite{GinibreGKS} to prove GKS correlation inequalities which he applied instead to what I've called the Ising domination problem.  I won't describe his work in detail here (but see \cite{MSW, PTLG}) but will limit things to quoting the part of his major theorem relevant to our focus in this note.

\begin{theorem} [Wells\cite{WellsTh, MSW}] Let $d\mu$ be an even probability measure on $\bbR$ with compact support that is not a point mass at $0$.  Then there is a strictly positive number, $T_-(\mu)$, so that $\mu$ Ising doninates $b_S$ if and only if $S\le T_-$.  Moreover
\begin{equation}\label{2.4}
   S\le T_-\iff \forall_{n\in\bbN} \int_\bbR (x^2-S^2)^n\,d\mu(x) \ge 0
\end{equation}
\end{theorem}

I should mention I happened to look at a 1981 paper of Bricmont, Lebowitz and Pfister~\cite{BLP} that includes in an appendix a proof (with attribution to Wells) of Wells result about the existence of $T_->0$.

One consequence of the theorem is
  \begin{equation}\label{2.5}
  T_- \le \left(\int_\bbR x^2\,d\mu(x)\right)^{1/2}
  \end{equation}
It is an interesting question when one has equality in this inequality.  I call the measure \emph{canonical} if one does have it.  One of my few new results was to show that if $\mu_D$ is the probability of distribution of a the first component of a $D$-component vector uniformly distributed on the unit sphere $S^{D-1}$ in $\bbR^D$, then $\mu_D$ is canonical (it was clear from the quote from my paper with Aizenman that Wells had proven this for $D=2$ in his preprint).

I also computed for $0\le\lambda\le 1$, $T_-$ for the probability measure supported by the three points $\{0,\pm1\}$ given by
\begin{equation}\label{2.6}
  d\mu_\lambda = \tfrac{\lambda}{2}\left(\delta_1+\delta_{-1}\right)+(1-\lambda)\delta_0
\end{equation}
and found that
\begin{equation}\label{2.7}
  T_-(\lambda) = \left\{
                   \begin{array}{ll}
                     \sqrt{\lambda}, & \hbox{ if }\lambda\le\tfrac{1}{2} \\
                     \sqrt{\tfrac{1}{2}}, & \hbox{ if }\lambda\ge\tfrac{1}{2}
                   \end{array}
                 \right.
\end{equation}
In particular , this measure is canonical if and only if $\lambda\le\tfrac{1}{2}$.  This shows some measures are canonical and others are not.

\section{The Conjecture} \lb{s3}

Consider the measure $d\tilde{\mu}_{S}$ discussed earlier - the scaled spin $S$ Ising model of $2S+1$ values equally spaced between $-1$ and $1$. We have just seen that for $S=1$ ($\lambda=\tfrac{2}{3}$ in the above example), one has that $T_-=\sqrt{\tfrac{1}{2}}<\sqrt{\tfrac{2}{3}}= \left(\int_\bbR x^2\,d\tilde{\mu}_{S=1}(x)\right)^{1/2}$. So $T_-\ne\left({\jap{x^2}_\mu}\right)^{1/2}$ for spin $1$ so that measure is \emph{not} canonical!

But I quickly determined that one should expect equality in all other cases.  I did spin $\tfrac{3}{2}$ by hand and used Mathematica to compute $\jap{(x^2-a_S)^{2n+1}}_S$ where $a_S=\left(\int_\bbR x^2\,d\tilde{\mu}_{S}(x)\right)$ for $S=2,\tfrac{5}{2},3$ and $m=1,2,\dots,10$ and for $S=20$ and $m=1,\dots,5$ and found them all positive which leads to a natural conjecture
\begin{equation}\label{3.1}
  \jap{(x^2-a_S)^{2n+1}}_S\ge 0
\end{equation}

As explained earlier, because Wells domination implies Ising domination, one has that for pair interactions
\begin{equation}\label{3.2}
  T_c(S) \ge  T_-(S)^2  T_c\left(\tfrac{1}{2}\right)
\end{equation}
As it turns out, there is a result of this genre in the literature.   In 1969, Griffiths wrote a famous paper~\cite{GriffTrick} on obtaining spin $S$ Ising spins by ferromagnetically coupling $2S$ spin $\tfrac{1}{2}$ spins together which lead to GKS and Lee Yang for spin $S$ Ising systems. This is usually presented in terms of an elegant coupling discussed in the first part of the paper. Less attention is paid to the second part where he shows instead there is such a coupling in which $S$ of the spin $\tfrac{1}{2}$ spins are frozen together (for $S$ half an odd integer, it's $S+\tfrac{1}{2}$) which he noted implies
\begin{equation}\label{3.3}
   T_c(S) \ge \tfrac{1}{4}  T_c\left(\tfrac{1}{2}\right)
\end{equation}

The quantity $a_S=\left(\int_\bbR x^2\,d\tilde{\mu}_{S}(x)\right)$ of \eqref{3.1} is equal to $\tfrac{1}{3}+\tfrac{1}{3S}$. If one proves that this is $T_-^2$ for $S\ne 1$, one has for such $S$ that
\begin{equation}\label{3.4}
  T_c(S) \ge \left(\frac{1}{3}+\frac{1}{3S}\right)  T_c\left(\tfrac{1}{2}\right)
\end{equation}
while for $S=1$ where we know that one has that $T_-^2=\tfrac{1}{2}$
\begin{equation}\label{3.5}
  T_c(1) \ge \frac{1}{2} T_c\left(\tfrac{1}{2}\right)
\end{equation}
Not only is this an improvement of Griffiths by more than $\tfrac{4}{3}$ but in the result for $S\ne 1$, the improved constant is optimal!! For one has equality if $T_c$ is replaced by its mean field values and, as noted by Dyson, Lieb and Simon~\cite{DLS2}, mean field theory is exact in the nearest neighbor infinite dimension limit.

Rescaling so the maximum spin value is $S$, the conjecture, \eqref{3.1}, is the assertion that for $m=1,2,\dots$ and $S =  \tfrac{3}{2}, 2, \tfrac{5}{2}, 3, \dots$, one has that
\begin{equation}\label{3.6}
  \sum_{j=-S}^{S}(3j^2-S(S+1))^{2m+1}\ge 0
\end{equation}
For $S$ an integer, this is the usual kind of sum.  For $2S$ an odd integer, $j$ takes the $2S+1$ values $-S,-S+1,\dots,S-1,S$, i.e. $2j$ is an odd integer.  Note, the constant $S(S+1)$ is such that the sum is zero if $m=0$

I found this conjecture fascinating and worked on it with no progress for about 7 months. I even got three of my coauthors from other papers to think about it with no luck.

Given that Lieb has a celebrated paper~\cite{LiebClassLim} on comparing Heisenberg models (admittedly classical vs. quantum and pressures, not correlations) and that I didn't want to bury in a long book this material which had already been buried for 45 years, it seemed natural to use this for an article when I was asked to contribute to a Festschrift for Elliott's $90^{th}$ birthday.  The paper was due on Jan 31, 2022 and on Friday, Jan 14, I had a first draft of the paper.

It seemed a shame not to make one more push to prove the conjecture so I did the obvious thing. Desperate situations call for desperate measures.  At 11 AM on Friday, Jan 14, I sent an email entitled \emph{``A challenge''} stating the conjectured inequality (and with the draft to explain its significance) to Terry Tao.  When I logged on after Shabbat the next evening I had an email dated Saturday at 1:30 PM with a proof of the conjecture!!!

But the scenario isn't quite as you image it.  At 1:30 PM on Friday, Terry had emailed me back: \emph{``I have a postdoc who works on some other inequalities vaguely of this type, I will forward this problem to him and see if he is interested.''} and it was the postdoc, Jos\'{e} Madrid, who sent the proof.

\section{Majorization} \lb{s4}

Jos\'{e}'s note had one wonderful idea (using Karamata's inequality) and 5 dense pages of calculation to implement it.  We Zoomed several times, first for me to offer him a coauthorship (Terry had suggested an appendix) and to discuss simplifying the implementation.  We discovered what we thought was a new criteria for majorization that led to a three line proof. OK, a proof with three long lines.  We then discovered that the proof was only really simple in case $S$ was half an odd integer.  As I'll explain, the integer case is harder but we found a proof in that case that was only a little longer.

As indicated, the key notion is majorization, a set of ideas that go back to Schur~\cite{Schur} in 1923 and  Hardy-Littlewood-P\'{o}lya~\cite{HLP1}.  A standard reference is Marshall-Olkin~\cite{MarOlk} which has been called a love poem to majorization; other references are Hardy-Littlewood-P\'{o}lya~\cite{HLP2} and Simon \cite[Chapters 14-15]{SimonConvex}.  I suspect my coauthors hadn't seen this theory but I didn't have this excuse. My convexity book has a whole chapter on it!

If $\mathbf{x}\in\bbR^n_{+}$ (the set with $x_1, x_2\,\dots x_n \ge 0$), we define, $\mathbf{x}^*$, its \emph{decreasing rearrangement} to be the point in $\bbR^n_+$ whose coordinates are a permutation of those of $\mathbf{x}$ with $x_1^*\ge x_2^*\ge\dots\ge x_n^*$.  If $\mathbf{x}, \mathbf{y}\in\bbR^n_{+}$ we say that $\mathbf{x}$ \textit{majorizes} $\mathbf{y}$, written $\mathbf{x}\succ \mathbf{y}$ or $\mathbf{y}\prec \mathbf{x}$ if an only if
\begin{equation}\label{4.1}
  \sum_{j=1}^{n} x^*_j = \sum_{j=1}^{n} y^*_j; \quad S_k(\mathbf{x})\equiv\sum_{j=1}^{k} x^*_j \ge \sum_{j=1}^{k} y^*_j,\,  k=1,\dots,n-1
\end{equation}
which defines $S_k(\mathbf{x})$.  Given $\pi\in\Sigma_n$ the group of permutations of $\{1,\dots,n\}$ and $\mathbf{x}\in\bbR^n_+$, one defines $\pi^*(\mathbf{x})$ by
\begin{equation}\label{4.2}
  \pi^*(\mathbf{x})_j = x_{\pi^{-1}(j)}
 \end{equation}

Majorization is the basis of a number of inequalities sometimes called rearrangement inequalities.  Basic to most of them is

\begin{proposition} \label{P4.1} $\mathbf{y}\prec \mathbf{x}$ if and only if $\mathbf{y}$ is in the convex hull in $\bbR^n$ of the (at most) n! points $\{\pi^*(\mathbf{x})\}_{\pi\in\Sigma_n}$.
\end{proposition}

This is proven by slicing $\bbR^n$ with specific hyperplanes; see Simon~\cite[Theorem 1.9]{SimonTrace} or Simon~\cite[Theorem 15.5]{SimonConvex}).  An immediate consequence is

\begin{theorem} [Karamata's Inequality\cite{Kara}] \label{T4.2} Let $\mathbf{x}, \mathbf{y}\in\bbR^n_{+}$ with $\mathbf{x}\succ \mathbf{y}$ and let $\varphi$ be an arbitrary continuous convex function on $[0,\max_j(x_j)]$.  Then
\begin{equation}\label{4.3}
 \sum_{j=1}^{n} \varphi(x_j) \ge \sum_{j=1}^{n} \varphi(y_j)
\end{equation}
\end{theorem}

Even though this is widely referred to as Karamata's inequality after Karamata's 1932 paper~\cite{Kara}, it or theorems that imply it appear in a 1923 paper of Schur~\cite{Schur} and a 1929 paper of Hardy-Littlewood-P\'{o}lya~\cite{HLP1}. That said, we note that HLP~\cite{HLP1} doesn't have a proof which may not have appeared until their 1934 book~\cite{HLP2} and that Karamata proved a converse, namely, if $\mathbf{x}, \mathbf{y}\in\bbR^n_{+}$ and the inequality holds for all convex $\varphi$, then $\mathbf{x}\succ \mathbf{y}$; see Simon~\cite[Theorem 15.5]{SimonConvex}).

The proof of Karamata's theorem is simple. One notes the function $\mathbf{w}\mapsto \sum_{j=1}^{n} \varphi(w_j)$ is convex and permutation symmetric and then uses Proposition \ref{P4.1}.

Madrid and I found a simple criterion for majorization.  Given the vast literature on the subject, we suspected it was already known but couldn't find it before \cite{MSW} was published although after I gave the talk on this work, we were told that Astashkin et. al~\cite{ALM} had a continuum analog a few months earlier (and it may well appear even earlier somewhere!).  Here is our criterion:

\begin{proposition} \lb{P4.3} Suppose that $\mathbf{x}, \mathbf{y}\in\bbR^n_{+}$ with $\sum_{j=1}^{n} x_j = \sum_{j=1}^{n} y_j$ and that for some $\ell\in \{2,\dots,n-1\}$, one has that
\begin{equation}\label{4.4}
  j< \ell \Rightarrow x^*_j > y^*_j \qquad \qquad j \ge \ell \Rightarrow x^*_j \le y^*_j
\end{equation}
Then $\mathbf{x}\succ \mathbf{y}$.
\end{proposition}

\begin{proof} Without loss, we can suppose that $x=x^*$, $y=y^*$. If $k< \ell$, it is immediate that $\sum_{j=1}^{k} x_j \ge \sum_{j=1}^{k} y_j$ and similarly, it is immediate that if $k\ge\ell$, then  $\sum_{j=k}^{n} x_j \le \sum_{j=k}^{\ell} y_j$.  Subtracting this from $\sum_{j=1}^{n} x_j = \sum_{j=1}^{n} y_j$, we see that also for $k\ge\ell$, one has that $\sum_{j=1}^{k} x_j \ge \sum_{j=1}^{k} y_j$.
\end{proof}

With this result in hand, we can turn to the proof of \eqref{3.6} at least when $S$ is half an odd integer.  Let us begin by explaining the difference between this case and the case of $S$ integral.  The point is that as $j$ runs through allowed values, there are degeneracies because $(-j)^2=j^2$.  When $2S$ is odd, the total number of points is even, $j=0$ is not allowed and every value occurs with multiplicity $2$ so all values with the same weight.  When $S$ is integral, $0$ is an allowed value of $j^2$ which has half the weight of every other allowed value of $j^2$ and this complicates the analysis.  Here is the $2S$ odd result and its proof:

\begin{theorem} [Madrid, Simon and Wells~\cite{MSW}] \label{T4.4} Fix an integer $N\ge 1$, a function, $\psi$ on $[0,1]$, which is non-negative, continuous, strictly monotone increasing and convex and a function, $\Phi$, on $[-\Norm{\psi}_\infty,\Norm{\psi}_\infty]$ which is continuous, odd and whose restriction to $[0,\Norm{\psi}]$ is convex.  Let
\begin{equation}\label{4.5}
  \overline{\psi} = (N+1)^{-1}\sum_{j=0}^{N} \psi\left(\tfrac{j}{N}\right)
\end{equation}
Then
\begin{equation}\label{4.6}
  \sum_{j=1}^{N} \Phi\left(\psi\left(\tfrac{j}{N}\right)-\overline{\psi}\right) \ge 0
\end{equation}
\end{theorem}

\begin{remark} By translation and scaling, this result can easily be generalized.  For example, while stated for $N+1$ equally spaced points between $0$ and $1$, we will apply it to $N+1$ half odd integers stating at $\tfrac{1}{2}$, i.e. $\tfrac{1}{2}, \tfrac{3}{2}, \dots, N+\tfrac{1}{2}$.  The map $k\mapsto (k-\tfrac{1}{2})/N$ maps those $N+1$ half odd integers into the points of the theorem.  Taking into account that the sum in \eqref{3.6} is twice the sum from $\tfrac{1}{2}$ to $S$, we see that because $j^2$ maps to a non-negative, continuous, strictly monotone increasing and convex  function under $k\mapsto (k-\tfrac{1}{2})/N$ and $u\mapsto u^{2m+1}$ is continuous, odd and whose restriction to $[0,\infty)$ is convex, that we have the Corollary below.
\end{remark}

\begin{corollary} [Madrid, Simon and Wells~\cite{MSW}] \label{C4.5} \eqref{3.6} holds for for $m=1,2,\dots$ and $S =\tfrac{1}{2},\tfrac{3}{2},\tfrac{5}{2},\tfrac{7}{2}\dots$.
\end{corollary}

\begin{remark} (3.6) continues to hold if $u^{2m+1}$ is replaced by any function which is continuous, odd and whose restriction to $[0,\infty)$ is convex and if $j^2$ is replaced by any even, non-negative, continuous, function whose restriction to $[0,\infty)$ is strictly monotone increasing and convex.
\end{remark}

We need two preliminaries for the proof of Theorem \ref{T4.4}:

\begin{lemma} \lb{L4.6} Let $\psi$ be a convex function on $[0,1]$ and suppose that
\begin{equation}\label{4.7}
  0\le \tilde{b}\equiv 2c-b < \tilde{a} \equiv 2c-a \le c \le a < b \le 1
\end{equation}
Then
\begin{equation}\label{4.8}
  \tfrac{1}{2}(\psi(b)+\psi(\tilde{b})) \ge \tfrac{1}{2}(\psi(a)+\psi(\tilde{a}))\ge \psi(c)
\end{equation}
Moreover, the first inequality is strict unless $\psi'(s)$ is constant on $(\tilde{a},a)$.
\end{lemma}

\begin{proof}  If one takes $a=c$ and then replaces $b$ by $a$, the first inequality becomes the second so it suffices to prove the first one.  Without loss (by translation and scaling) we can take $c=\tfrac{1}{2}, b=1$ so that $\tilde{b}=0$ and $\tilde{a}=1-a$. By the fundamental theorem of calculus (a general convex function is not $C^1$ but it is differentiable with the possible exception of a countable set and the fundamental theorem of calculus holds; see Simon \cite[Theorem 1.28]{SimonConvex})
\begin{equation}\label{4.9}
  \tfrac{1}{2}(\psi(1)+\psi(0))-\tfrac{1}{2}(\psi(a)+\psi(\tilde{a})) = \tfrac{1}{2} \int_{a}^{1} [\psi'(s)-\psi'(1-s)]\, ds
\end{equation}
By convexity, the integrand is non-negative so we have proven \eqref{4.8}.  Moreover if $\psi'(s)$ is not constant on $(1-a,a)$, then the integral is strictly positive.
\end{proof}

\begin{proposition} \lb{P4.7} Let $\psi$, $\overline{\psi}$ and $N$ be as in Theorem \ref{4.4}.  Then
\begin{equation}\label{4.10}
  n \equiv \#\{j\,\mid\,\psi\left(\tfrac{j}{N}\right) \le \overline{\psi}\} \ge (N+1)/2
\end{equation}
and
\begin{equation}\label{4.11}
  \psi\left(\tfrac{1}{2}\right) \le \overline{\psi} \le \tfrac{1}{2}(\psi(0)+\psi(1))
\end{equation}
Moreover, the inequalities in \eqref{4.11} are strict if $N\ge 2$ and $\psi$ is not an affine function on $[0,1]$ (i.e. $\psi'$ is not constant).
\end{proposition}

\begin{proof} For any $j=0,1,\dots,N$, \eqref{4.8} implies that
\begin{equation}\label{4.12}
   \psi\left(\tfrac{1}{2}\right) \le  \tfrac{1}{2}\left(\psi\left(\tfrac{j}{N}\right)+\psi\left(1-\tfrac{j}{N}\right)\right) \le \tfrac{1}{2}(\psi(0)+\psi(1))
\end{equation}
Averaging over $j$ yields \eqref{4.11}.  If $\psi$ is not affine on $[0,1]$, then the second inequality is strict for $1\le j \le N-1$ so the second inequality in \eqref{4.11} is strict.  Since $\psi\left(\tfrac{1}{2}\right) < \tfrac{1}{2}(\psi(0)+\psi(1))$ if $\psi$ is not affine, we see that in the case the first inequality is always strict.

Since $\psi$ is strictly monotone, the first inequality in \eqref{4.11} implies the unique $x\in [0,1]$ with $\psi(x)=\overline{\psi}$  has $x\ge\tfrac{1}{2}$. This implies that $n = \#\{j\,\mid\,\tfrac{j}{N} \le x\} \ge \#\{j\,\mid\,\tfrac{j}{N} \le \tfrac{1}{2}\} \ge (N+1)/2.$
\end{proof}

\begin{proof} [Proof of Theorem \ref{T4.4}]  Let $q=N+1-n\le n$ by \eqref{4.9}.  Define
\begin{equation}\label{4.13}
  y_j= \overline{\psi} - \psi\left(\tfrac{j-1}{N}\right) \qquad j=1,\dots,n
\end{equation}
\begin{equation}\label{4.14}
  x_j = \left\{
          \begin{array}{ll}
            \psi\left(\tfrac{N+1-j}{N}\right) - \overline{\psi}\, & \hbox{ if } j=1,\dots,q \\
            0, & \hbox{ if } j \ge q
          \end{array}
        \right.
\end{equation}
Since $\psi$ is monotone and $n$ is defined by \eqref{4.10}, we have that $\mathbf{x}, \mathbf{y}\in\bbR^n_{+,\ge}=\{\mathbf{x}\in\bbR^n_+\,\mid\,x_1\ge\dots,\ge x_n\}$.  By the definition of $\overline{\psi}$, we have that
\begin{equation}\label{4.15}
  \sum_{j=1}^{n} x_j = \sum_{j=1}^{n} y_j
\end{equation}

If $N=1$ or $\psi$ is affine on $[0,1]$, it is easy to see that $x_j=y_j$ for all $j$, so, since $\Phi$ is odd, we have that \eqref{4.1} holds.  Thus henceforth we will suppose that  $N\ge 2$ and $\psi$ is not an affine function on $[0,1]$, so, in particular, the inequalities in \eqref{4.11} are strict.

Note next that because $\psi$ is assumed convex, we have that
\begin{equation}\label{4.16}
  m<p \Rightarrow \psi\left(\tfrac{m+1}{N}\right)-\psi\left(\tfrac{m}{N}\right) \le \psi\left(\tfrac{p+1}{N}\right)-\psi\left(\tfrac{p}{N}\right)
\end{equation}

By the strict form of \eqref{4.11}, $x_1>y_1$.  Because of \eqref{4.15}, there must be a first $\ell$ so that $x_\ell\le y_\ell$.  We claim that if $\ell<n$, then $x_{\ell+1}\le y_{\ell+1}$.  If $\ell+1>q$, then $x_{\ell+1}=0$ and the required inequality is immediate.  If $\ell+1\le q$, then \eqref{4.16} implies that $x_{\ell}-x_{\ell+1} \ge y_\ell-y_{\ell+1}$.  Subtracting this from  $x_\ell\le y_\ell$ proves that $x_{\ell+1}\le y_{\ell+1}$.  Repeating this argument, proves that for all $j\ge\ell$ we have that $x_j\le y_j$.  Thus by Proposition \ref{P4.3}, $\mathbf{x}\succ \mathbf{y}$.

By Karamata's inequality, \eqref{4.3}, we conclude that $\sum_{j=1}^{n} \Phi(x_j)-\Phi(y_j) \ge 0$.  Since $\Phi$ is odd and $\Phi(0)=0$, this is equivalent to \eqref{4.6}.
\end{proof}

\begin{example} \lb{E4.8} To understand why we need the extra condition to handle the case when $S$ is integral, consider $d\tilde{\mu}_S$ for $S=6$ scaled to have spacing $1$, i.e. $13$ pure points with weight $1/13$ at $0,\pm 1,\pm 2,\pm 3,\pm 4,\pm 5,$ $\pm 6$.  The average of the square is $14$.  The values of $j^2$ are $j^2=0,1,1,4,4,9,9,16,16,25,25,36,36$ so $n=7$ values are less than $14$ and one sees that (ignore $\mathbf{w}$ for now)
\begin{align}\label{4.17}
  \mathbf{x} &= 22,22,11,11,\,\;2,2,0 \nonumber \\
  \mathbf{y} &= 14,13,13,10,10,5,5 \nonumber \\
  \mathbf{w} &= 22,22,\,\;0,11,11,2,2
\end{align}
One can verify that $\mathbf{x}\succ \mathbf{y}$ by hand (and, below, we will prove the result for all $S\ge 2$) but one can't use Proposition \ref{P4.3} as we did in our proof of Theorem \ref{T4.4} for $x_j-y_j$ shifts signs three times instead of one time.  The problem is that the components of $\mathbf{x}$ and $\mathbf{y}$ are paired but shifted.

Look at $\mathbf{w}$ which we get by moving the $0$ from position $7$ to position $3$.  One can handle the first  three partial sums by noting that $22+22\ge 14+13+13$ and the remaining partial sums by noting that there is only one sign change after the third place and use Proposition \ref{P4.3} to prove the partial sums of $\mathbf{w}$ dominate those of $\mathbf{y}$ and note it is trivial that partial sums of $\mathbf{x}$ dominate those of $\mathbf{w}$.  The key is that by moving the $0$, the pairs are no longer shifted.
\end{example}

We will need an extra condition that implies as in this example, the partial sum of the first two $x_j$'s dominates the partial sum of the first three $y_k$'s.  The general theorem analogous to Theorem \ref{T4.4} is

\begin{theorem} [Madrid, Simon and Wells~\cite{MSW}] \lb{T4.9} Fix an integer $N\ge 2$ and an even, continuous, convex function, $\psi$ on $[-1,1]$ and a function, $\Phi$, on $[-\Norm{\psi}_\infty,\Norm{\psi}_\infty]$ which is continuous, odd and whose restriction to $[0,\Norm{\psi}]$ is convex.  Let $\overline{\psi}$ be given by
\begin{equation}\label{4.18}
  \overline{\psi} = (2N+1)^{-1}\sum_{j=-N}^{N} \psi\left(\tfrac{j}{N}\right)
\end{equation}
Suppose that
\begin{equation}\label{4.19}
  2\psi(1)+\psi(0)+2\psi\left(\tfrac{1}{N}\right) \ge 5 \overline{\psi}
\end{equation}
If $N$ is odd, suppose that
\begin{equation}\label{4.20}
  \psi\left(\tfrac{1}{2}+\tfrac{1}{2N}\right) \le \overline{\psi}
\end{equation}
Then
\begin{equation}\label{4.21}
  \sum_{j=-N}^{N} \Phi\left(\psi\left(\tfrac{j}{N}\right)-\overline{\psi}\right) \ge 0
\end{equation}
\end{theorem}

\begin{remarks} 1. The condition \eqref{4.19} is exactly the condition that the partial sum of the first $2$ $x_j$'s dominates the partial sum of the first three $y_k$'s.

2. It might be true that this theorem holds without the need for the condition \eqref{4.19} but it holds in the case we need so MSW didn't try hard to eliminate it.  The example above shows why a naive extension of the proof of Theorem \ref{T4.4} doesn't work and led to the extra condition.  We do note that \eqref{4.19} is a restriction.  If we normalize $\psi$ by $\psi(0)=0, \psi(1)=1$, then in the limit as $N\to\infty$, \eqref{4.19} becomes
\begin{equation}\label{4.22}
  \int_{0}^{1} \psi(x)\,dx \le \tfrac{2}{5}
\end{equation}
which for $\psi(x)=|x|^p$ requires $p\ge \tfrac{3}{2}$ while convexity only requires $p\ge 1$.

3. On the other hand, \eqref{4.20} is quite natural independent of our method of proof; see Madrid, Simon and Wells~\cite{MSW}.
\end{remarks}

The reader can check Madrid, Simon and Wells~\cite{MSW} for the proof but we note it combines the ideas of the proof of Theorem \ref{T4.4} and Example \ref{E4.8}. We should give the details of checking that the Theorem proves \eqref{3.6} for $S=2,3,4,\dots$.  One needs to check \eqref{4.18} and \eqref{4.20}, the latter for $S (=N)$ odd and $S\ge 3$.  These are
\begin{align}
  2S^2+2 &\ge \tfrac{5}{3}S(S+1); \,S=2,3,4,\dots \label{4.23} \\
  S^2(\tfrac{1}{2}+\tfrac{1}{2S})^2 &\le \tfrac{1}{3}S(S+1); \,S=3,5,7,\dots \label{4.24}
\end{align}
\eqref{4.23} is equivalent to $0\le S^2-5S+6=(S-2)(S-3)$ which holds for all integral $S$.  \eqref{4.24} is equivalent to $0\le S^2-2S-3=(S-3)(S+1)$ which holds for all $S\ge 3$.

\section{The End of the Story} \lb{s5}

In our first Zoom call, Jos\'{e} also suggested it would be good to try again to locate Daniel Wells. I wasn't starting at ground zero.  While I got nothing from Indiana University, I talked about this material during the conference in honor of my $75^{th}$ birthday and Leonard Schulman, a computer scientist at Caltech (and son of a student of Arthur Wightman), heard my talk and did some Google searching.  He found a short story available via Kindle on Amazon whose About the Author read

\emph{Daniel R Wells was born in Sterling, Illinois on March 15, 1945. He attended the local parochial schools and graduated from high school in 1963. In October of that year he enlisted in the United States Navy and served for four years. After the Navy, he started college in 1968, studying mathematics, eventually earning a PhD from Indiana University in 1977. He taught mathematics for two years at Texas A\&M and then returned to school at the University of Illinois to study computer science. He achieved a PhD in 1982 and worked for various companies as a software engineer until he retired in 2004.}

I wasn't clever enough to pull on the right threads of this fabric.  Since I had friends at Texas A\&M, I consulted them to see if they could find any record.  Nope.  I tried to leave a ``review'' of his book saying I wanted to contact the author about his thesis but Amazon said it wasn't a review and wouldn't post it.  I bought his Kindle book hoping it might provide more information but it didn't. What I should have done is contact U of I computer science where he got his second PhD. and where he has continued to do some teaching.

Spurred by Jos\'{e}, I posted a message on Facebook where I have a group of friends mainly mathematicians and theoretical physicists. The message gave some background and asked if anyone had any idea how to follow up.  Joshua Paik, a math grad student at Penn State told me he regarded himself as an internet sleuth. The next morning I had a link in a private message from Mr. Paik to a Find a Person internet site with the right name, the right age who lived in the town where the Amazon profile said Wells was born.  Shortly after that, Mr. Paik sent me what he though might be Wells' email address. I contacted the email address asking if the recipient was a Daniel Wells who got a math PhD in Indiana then sent him the current draft and asked him to be a coauthor - after all, 2.5 out of 6 sections were from his thesis! He agreed, so in less than a week, I picked up two coauthors.

The next week, Jos\'{e} and I zoomed with Daniel and I got some more background.  Wells had gone to Texas A\&M for a postdoc, written up his thesis with the addition of the rotor-Ising comparison theorem and sent the preprint that Aizenman and I referred to off to a journal where it should have been accepted.  But it was rejected. At this point, his thesis advisor should have stepped in and explained the facts of life: just as there are bad papers, there are bad referees and one should send the paper off to another journal.  But alas, Slim Sherman, his advisor, had passed away shortly before he took his oral exam and wasn't there to advise him. This was in the old days when postdocs didn't have formal advisors!  Wells was so discouraged, he totally left mathematics even though he'd written a very good thesis.  Sometimes the system doesn't work.

The moral is that young workers in particular need someone to say to them: \emph{Jerk, submit it to another journal}. Your reaction might be that I'm the wrong one to say that since surely I've never had any of my papers rejected.  But I've had many papers rejected.  I thought I'd close with two amusing stories of rejections so that you might even respond to rejection with some humor.

One involves my paper with Elliott Lieb~\cite{LSTF} whose main result was that Thomas-Fermi Theory was exact in a certain limit of large $Z $atoms and molecules.  It was significant enough that on its $35^{th}$ anniversary, some quantum chemists had a conference marking the occasion!  This announcement was originally rejected by Physical Review Letters with a report that paraphrased began \emph{This paper is one of the worst papers I have ever seen.  It is a sequence of unproven assertions}(true; it was an announcement after all) \emph{many of which are obviously wrong.  For example, the author assert that the Thomas-Fermi density is $C^\infty$ which would make it $0, 1$ or $\infty$ depending on the value of $C$.}  The referee was clearly unfamiliar with modern mathematical notation and incompetent to evaluate a paper in mathematical physics.  We complained and asked for a second referee who accepted it. (General, I do not recommend resubmitting to the same journal).

The other involves my paper with Christiansen and Zinchenko~\cite{CSZCheb1} which settled a $40$ year old conjecture of Widom~\cite{Widom} which was regarded as a major open question in the asymptotics of extremal polynomials.  I thought it was good enough that we should submit it to one of the top three math journals but for various reasons, one of my coauthors wanted it to go to a journal just below those three.  We got a quick rejection that paraphrased said \emph{It is nice to have a $40$ year old conjecture resolved but this paper should be rejected because the proof is too easy so it isn't up to the high standard of $<journal name>$}.  I was scandalized by this report so we submitted it after all to a top three journal where it was accepted!


\end{document}